\documentclass{ap-jnmp}

\def\diag{{\rm diag}}
\def\Mat{{\rm Mat}}
\def\G{{\rm G}}
\def\g{\frak{g}}
\def\T{{\rm T}}
\def\SU{{\rm SU}}
\def\U{{\rm U}}
\def\su{\frak{su}}
\def\C{{\mathbb C}}
\def\R{{\mathbb R}}
\def\SO{{\rm SO}}
\def\OO{{\rm O}}
\def\Spin{{\rm Spin}}
\def\tr{{\rm tr}}
\def\rank{{\rm rank}}

\copyrightauthor{}

\title{On constant solutions of $\SU(2)$ Yang-Mills equations with arbitrary current\\ in Euclidean space $\R^n$}

\author{Dmitry Shirokov}

\address{National Research University Higher School of Economics,\\
Myasnitskaya str. 20, Moscow, 101000, Russia\\
\email{dm.shirokov@gmail.com}}

\address{Institute for Information Transmission Problems of Russian Academy of Sciences,\\
Bolshoy Karetny per. 19, Moscow, 127051, Russia}

\begin{document}

\maketitle
\thispagestyle{empty}

\vphantom{\vbox{%
\begin{history}
\received{(Day Month Year)}
\revised{(Day Month Year)}
\accepted{(Day Month Year)}
\end{history}
}}

\begin{abstract}
In this paper, we present all constant solutions of the Yang-Mills equations with $\SU(2)$ gauge symmetry for an arbitrary constant non-Abelian current in Euclidean space $\R^n$ of arbitrary finite dimension $n$. Using the invariance of the Yang-Mills equations under the orthogonal transformations of coordinates and gauge invariance, we choose a specific system of coordinates and a specific gauge fixing for each constant current and obtain all constant solutions of the Yang-Mills equations in this system of coordinates with this gauge fixing, and then in the original system of coordinates with the original gauge fixing. We use the singular value decomposition method and the method of two-sheeted covering of orthogonal group by spin group to do this. We prove that the number ($0$, $1$, or $2$) of constant solutions of the Yang-Mills equations in terms of the strength of the Yang-Mills field depends on the singular values of the matrix of current. The explicit form of all solutions and the invariant $F^2$ can always be written using singular values of this matrix. The relevance of the study is explained by the fact that the Yang-Mills equations describe electroweak interactions in the case of the Lie group $\SU(2)$. Nonconstant solutions of the Yang-Mills equations can be considered in the form of series of perturbation theory. The results of this paper are new and can be used to solve some problems in particle physics, in particular, to describe physical vacuum and to fully understand a quantum gauge theory.
\end{abstract}

\ccode{2000 Mathematics Subject Classification: 70S15}

\section{Introduction}

Up to now the law of elementary particles physics is given by quantum gauge theories \cite{Fad}. We need exact solutions of classical Yang-Mills equations to describe the vacuum structure of the theory and to fully understand a quantum gauge theory \cite{nian}. During the last 50 years, many scientists have been searching for particular classes of solutions of the Yang-Mills equations. The well-known classes of solutions of the Yang-Mills equations are described in detail in various reviews \cite{Actor}, \cite{Zhdanov}. Only certain (nontrivial) classes of particular solutions of these equations are known because of their nonlinearity: monopoles \cite{WYa}, \cite{tH}, \cite{Pol}, instantons \cite{Bel}, \cite{Wit}, \cite{ADHM}, merons \cite{deA}, etc.

The main result of this paper is the presentation of all constant (that do not depend on $x\in\R^n$) solutions of the Yang-Mills equations with $\SU(2)$ gauge symmetry for an arbitrary constant non-Abelian current in Euclidean space of arbitrary finite dimension. The relevance of the study is explained by the fact that the Yang-Mills equations describe electroweak interactions in the case of the Lie group $\SU(2)$. Note that instantons are solutions in Euclidean space-time (with imaginary time) and thus the Euclidean case (not only Minkowski case) is important for applications. Constant solutions of the Yang-Mills equations are essentially nonlinear solutions and, from this point of view, are particularly interesting for applications.

Constant solutions of the Yang-Mills equations with zero current were considered in \cite{Sch} and \cite{Sch2}. In \cite{Sch}, Prof. R. Schimming wrote: \emph{``The following problems concerning constant Yang-Mills fields are actual ones in our opinion: Is there a gauge- and coordinate-invariant characterization of those Yang-Mills fields which admit constant potentials with respect to some gauge and some coordinate system? Find as many as possible (in the ideal case: all) constant Yang-Mills fields and classify them! \ldots''} In the current paper, we give a complete answer to these questions in the case of the Lie group $\SU(2)$. Our results for an arbitrary current are consistent with the results of \cite{Sch} and \cite{Sch2} for zero current (and arbitrary compact Lie algebra).

In this paper, we present the general solution of the special system (system of the $\SU(2)$ Yang-Mills equations for constant solutions with arbitrary current) of $3n$ cubic equations with $3n$ unknowns and $3n$ parameters. This algebraic problem is solved using the singular value decomposition method and the method of two-sheeted covering of orthogonal group by spin group. Using the invariance of the Yang-Mills equations under the orthogonal transformations of coordinates and gauge invariance, we choose a specific system of coordinates and a specific gauge fixing for each constant current and obtain all constant solutions of the Yang-Mills equations in this system of coordinates with this gauge fixing, and then in the original system of coordinates with the original gauge fixing.

\section{The main ideas}

Let us consider Euclidean space $\R^n$ of arbitrary finite dimension $n$. We denote Cartesian coordinates by $x^\mu$, $\mu=1, \ldots, n$ and partial derivatives by $\partial_\mu=\partial/{\partial x^\mu}$.

Let us consider the Lie group
\begin{eqnarray}
\G=\SU(2)=\{S\in\Mat(2,\C) \,| \, S^\dagger S= {\bf 1}, \det S=1\},\qquad \dim G=3
\end{eqnarray}
and the corresponding Lie algebra
\begin{eqnarray}
\g=\su(2)=\{S\in\Mat(2,\C) \,| \,S^\dagger=-S, \tr S=0\},\qquad \dim\g=3.
\end{eqnarray}

Denote by $\g\T^a_b$ a set of tensor fields of $\R^n$ of type $(a,b)$ with values in the Lie algebra $\g$. The metric tensor of $\R^n$ is given by the identity matrix
${\bf 1}=\diag(1,\ldots,1)=\|\delta_{\mu\nu}\|=\|\delta^{\mu\nu}\|$. We can raise or lower indices of components of tensor fields with the aid of the metric tensor. For example, $F^{\mu\nu}=\delta^{\mu\alpha}\delta^{\nu\beta}F_{\alpha\beta}$.

Let us consider the Yang-Mills equations
\begin{eqnarray}
&&\partial_\mu A_\nu-\partial_\nu A_\mu- [A_\mu,A_\nu]=F_{\mu\nu},\label{YM1}\\
&&\partial_\mu F^{\mu\nu}- [A_\mu,F^{\mu\nu}]=J^\nu,\label{YM2}
\end{eqnarray}
where $A_\mu\in \g\T_1$ is the potential, $J^\nu\in \g\T^1$ is the non-Abelian current, $F_{\mu\nu}=-F_{\nu\mu}\in \g\T_{2}$ is the strength of the Yang-Mills field. One suggests that $A_\mu, F_{\mu\nu}$ are unknown and $J^\nu$ is known.

Note that (\ref{YM1}) can be considered as a definition of the strength
\begin{eqnarray}
F_{\mu\nu}:=\partial_\mu A_\nu-\partial_\nu A_\mu- [A_\mu,A_\nu].\label{str}
\end{eqnarray}
We can substitute the components of the skew-symmetric tensor $F^{\mu\nu}$ from (\ref{YM1}) into (\ref{YM2}) and obtain
\begin{equation}
\partial_\mu(\partial^\mu A^\nu-\partial^\nu A^\mu- [A^\mu,A^\nu])- [A_\mu,\partial^\mu A^\nu-\partial^\nu A^\mu-[A^\mu,A^\nu]]=J^\nu.\label{YM:A}
\end{equation}
From this point of view, the system (\ref{YM1}) - (\ref{YM2}) can be considered as a system for the uknown $A_\mu$ and the known $J_\nu$. For each potential $A_\mu$, the corresponding strength (\ref{str}) can be calculated. Note that from a physical point of view, strength is important, not potential. Also note that a special case of the system (\ref{YM1}) - (\ref{YM2}) for the Abelian Lie group $\G=\U(1)$ (Maxwell's equations) can be considered only for the unknown $F_{\mu\nu}$ and the known $J_\nu$, but in this paper, we consider the case of the non-Abelian Lie group $\G=\SU(2)$ and need the potential $A_\mu$ for calculations.

We may verify that the current (\ref{YM2}) satisfies the non-Abelian conservation law
\begin{equation}
\partial_\nu J^\nu-[A_\nu,J^\nu]=0.
\end{equation}

The Yang-Mills equations are gauge invariant. Namely, the transformed tensor fields
\begin{eqnarray}
\acute A_\mu &=& S^{-1}A_\mu S-S^{-1}\partial_\mu S,\nonumber\\
\acute F_{\mu\nu} &=& S^{-1}F_{\mu\nu}S,\qquad S=S(x):\R^n\to\G, \label{gauge:tr}\\
\acute J^{\nu} &=& S^{-1}J^{\nu}S\nonumber
\end{eqnarray}
satisfy the same equations
\begin{eqnarray*}
&&\partial_\mu\acute A_\nu-\partial_\nu
\acute A_\mu-[\acute A_\mu,\acute A_\nu]=\acute F_{\mu\nu},\\
&&\partial_\mu\acute F^{\mu\nu}-[\acute A_\mu,\acute F^{\mu\nu}]=
\acute J^\nu.
\end{eqnarray*}
One says equations (\ref{YM1}) - (\ref{YM2}) are gauge invariant w.r.t. the transformations (\ref{gauge:tr}). The Lie group $\G$ is called the gauge group of the Yang-Mills equations.

From (\ref{YM:A}), we obtain the following algebraic system of equations for constant solutions (that do not depend on $x\in\R^n$)
\begin{eqnarray}
[A_\mu,[A^\mu, A^\nu]]=J^\nu,\qquad \nu=1, \ldots, n,\label{YMconst}
\end{eqnarray}
and the following expression for the strength of the Yang-Mills field
$$F^{\mu\nu}=-[A^\mu, A^\nu].$$

Constant solutions of the Yang-Mills equations with zero current $J^\mu=0$ were considered in \cite{Sch} and \cite{Sch2}. In this paper, we give all solutions of (\ref{YMconst}) for an arbitrary constant non-Abelian current $J^\nu$, $\nu=1, \ldots, n$.

Note that we have already studied constant solutions of the Yang-Mills-Proca equations, which generalize the Yang-Mills equations and the Proca equation, in \cite{YMP} and covariantly constant solutions of the Yang-Mills equations in \cite{YMM}, \cite{PFE}, \cite{YMSh} using the techniques of Clifford algebras. We do not use these results in the current paper.

Our aim is to obtain a general solution $A^\mu\in\su(2)\T^1$ of (\ref{YMconst}) for any $J^\mu\in\su(2)\T^1$. If $n=1$, then (\ref{YMconst}) transforms into $0=J^1$. Therefore, the equation (\ref{YMconst}) has an arbitrary solution $A^1\in \g\T^1$ for $J^1=0$ and it has no solutions for $J^1\neq 0$. Note that two dimensional Yang-Mills theory is discussed in \cite{GN} and other papers. We consider the case $n\geq 2$ further for the sake of completeness.

The Pauli matrices $\sigma^a$, $a=1, 2, 3$
\begin{eqnarray}
\sigma^1=\left( \begin{array}{cc} 0 & 1\\ 1 & 0 \end{array}\right),\quad
\sigma^2=\left( \begin{array}{cc} 0 & -i\\i & 0 \end{array}\right),\quad
\sigma^3=\left( \begin{array}{cc} 1 & 0\\ 0 & -1 \end{array}\right)
\end{eqnarray}
satisfy
$$(\sigma^a)^\dagger=\sigma^a,\qquad \tr\sigma^a=0,\qquad \{\sigma^a, \sigma^b\}=2\delta^{ab}{\bf 1},\qquad [\sigma^a, \sigma^b]=2i \epsilon^{ab}_{\,\,\,\,\,c}\sigma^c,$$
where $\epsilon^{ab}_{\,\,\,\,\,c}=\epsilon^{abc}$ is the antisymmetric Levi-Civita symbol, $\epsilon^{123}=1$.

We can take the following basis of the Lie algebra $\su(2)$:
\begin{eqnarray}
\tau^1=\frac{\sigma^1}{2i},
\qquad \tau^2=\frac{\sigma^2}{2i},
\qquad \tau^3=\frac{\sigma^3}{2i}
\end{eqnarray}
\begin{eqnarray}
\mbox{with}\qquad (\tau^a)^\dagger=-\tau^a,\qquad \tr\,\tau^a=0,\qquad [\tau^a, \tau^b]=\epsilon^{ab}_{\,\,\,\,\,c}\tau^c,
\end{eqnarray}
i.e. the structural constants of the Lie algebra $\su(2)$ are the Levi-Civita symbol.

For the potential and the current, we have
\begin{eqnarray}
A^\mu=A^\mu_{\,\,a} \tau^a,\qquad J^\mu=J^\mu_{\,\,a} \tau^a,\qquad A^\mu_{\,\, a}, J^\mu_{\,\, a}\in\R.\label{basisras}
\end{eqnarray}
Latin indices take values $a=1, 2, 3$ and Greek indices take values $\mu=1, 2, \ldots, n$.

Let us substitute (\ref{basisras}) into (\ref{YMconst}). We have
\begin{eqnarray}
&&A_{\mu c} A^\mu_{\,\,a} A^\nu_{\,\,b}[\tau^c,[\tau^a, \tau^b]]=J^\nu_{\,\,a} \tau^a,\nonumber\\
&&A_{\mu c} A^\mu_{\,\,a} A^\nu_{\,\,b} \epsilon^{ab}_{\,\,\,\,\,d}[\tau^c,\tau^d]=J^\nu_{\,\,a} \tau^a,\nonumber\\
&&A_{\mu c} A^\mu_{\,\,a} A^\nu_{\,\,b}\epsilon^{ab}_{\,\,\,\,\,d}\epsilon^{cd}_{\,\,\,\,\,k}\tau^k=J^\nu_{\,\,a} \tau^a,\nonumber
\end{eqnarray}
and, finally,
\begin{eqnarray}
A_{\mu c} A^\mu_{\,\,a} A^\nu_{\,\,b}\epsilon^{ab}_{\,\,\,\,\,d}\epsilon^{cd}_{\,\,\,\,\,k}=J^\nu_{\,\,k},\qquad \nu=1, \ldots, n,\qquad k=1, 2, 3.\label{eq}
\end{eqnarray}
We obtain $3n$ equations ($k=1, 2, 3$, $\nu=1, 2, \ldots, n$) for $3n$ expressions $A^\nu_{\,\,k}$ and $3n$ expressions $J^\nu_{\,\, k}$. We can consider this system of equations as the system of equations for two matrices $A_{n \times 3}=||A^\nu_{\,\, k}||$ and $J_{n \times 3}=||J^\nu_{\,\, k}||$.

We will give the general solution $A^\nu_{\,\,k}$ of the system (\ref{eq}) for all $J^\nu_{\,\, k}$ using algebraic methods. In Section 3, we also calculate the strength $F_{\mu\nu}$  (using (\ref{str})) and the invariant $F^2=F_{\mu\nu}F^{\mu\nu}$ for each solution $A_\mu$, because they are important from a physical point of view.

We have the following well-known theorem on the singular value decomposition (SVD), see \cite{SVD1}, \cite{SVD2}.
\begin{theorem}\label{thSVD}
For an arbitrary real matrix $A_{n\times N}$ of the size $n\times N$, there exist orthogonal matrices $L_{n\times n}\in\OO(n)$ and $R_{N\times N}\in\OO(N)$ such that
\begin{eqnarray}
L_{n\times n}^\T A_{n\times N}R_{N \times N}= D_{n\times N},\label{hs1}
\end{eqnarray}
where
$$
D_{n\times N}=\diag(\mu_1, \ldots, \mu_s),\qquad s=\min(n,N),\qquad \mu_1\geq \mu_2 \geq \cdots \geq \mu_s \geq 0.
$$
\end{theorem}
The numbers $\mu_1, \ldots, \mu_s$ are called the singular values, the columns $l_i$ of the matrix $L$ are called the left singular vectors, the columns $r_i$ of the matrix $R$ are called the right singular vectors. From (\ref{hs1}), we get $AR=LD$ and $A^\T L=RD^\T$. We obtain the following relation:
$$A A^\T L=LDR^\T R D^\T=L D D^\T,\qquad A^\T A R=R D^\T L^\T L D=R D^\T D,$$
i.e. the columns of the matrix $L$ are eigenvectors of the matrix $A A^\T$, and the columns of the matrix $R$ are eigenvectors of the matrix $A^\T A$. The squares of the singular values are the eigenvalues of the corresponding matrices. From this fact, it follows that singular values are uniquely determined.

\begin{lemma}\label{lemma1} The system of equations (\ref{eq}) is invariant under the transformation
$$A \to \acute{A}=AP,\qquad J \to \acute{J}=J P,\qquad P\in\SO(3)$$
and under the transformation
$$A \to \hat{A}=QA,\qquad J \to \hat{J}=Q J,\qquad Q\in\OO(n).$$
\end{lemma}
\begin{proof} The system (\ref{YMconst}) is invariant under the transformation
\begin{eqnarray}
\acute A_\mu = S^{-1}A_\mu S,\qquad \acute J^{\nu} = S^{-1}J^{\nu}S,\qquad S\in\G=\SU(2).\label{gauge:tr2}
\end{eqnarray}
It follows from the invariance under (\ref{gauge:tr}) and the fact that an element $S\in \G=\SU(2)$ does not depend on $x$ now.

Let us use the theorem on the two-sheeted covering of the orthogonal group $\SO(3)$ by the spin group $\Spin(3)\simeq\SU(2)$. For an arbitrary matrix $P=||p^a_b||\in\SO(3)$, there exist two matrices $\pm S\in\SU(2)$ such that
$$S^{-1}\tau^a S=p^a_b \tau^b.$$

We conclude that the system (\ref{eq}) is invariant under the transformation
\begin{eqnarray}
&&\acute{A^\mu}=S^{-1}A^\mu_{\,\,a} \tau^a S=A^\mu_{\,\,a} S^{-1}\tau^a S=A^\mu_{\,\,a} p^a_b \tau^b=\acute{A^\mu_{\,\,b}} \tau^b,\qquad \acute{A^\mu_{\,\,b}}=A^\mu_{\,\,a} p^a_b,\nonumber\\
&&\acute{J^\mu}=S^{-1}J^\mu_{\,\,a} \tau^a S=J^\mu_{\,\,a} S^{-1}\tau^a S=J^\mu_{\,\,a} p^a_b \tau^b=\acute{J^\mu_{\,\,b}} \tau^b,\qquad \acute{J^\mu_{\,\,b}}=J^\mu_{\,\,a} p^a_b.\nonumber
\end{eqnarray}

The Yang-Mills equations are invariant under the orthogonal transformation of coordinates. Namely, let us consider the transformation $x^\mu \to \hat{x^\mu}=q^\mu_\nu x^\nu$, where $Q=||q^\mu_\nu||\in\OO(n).$  The system (\ref{eq}) is invariant under the transformation
\begin{eqnarray}
&&\hat{A^\nu}=q^\nu_\mu A^\mu=q^\nu_\mu A^\mu_{\,\,a} \tau^a=\hat{A^\nu_{\,\,a}} \tau^a,\qquad \hat{A^\nu_{\,\,a}}=q^\nu_\mu A^\mu_{\,\,a},\nonumber\\
&&\hat{J^\nu}=q^\nu_\mu J^\mu=q^\nu_\mu J^\mu_{\,\,a} \tau^a=\hat{J^\nu_{\,\,a}} \tau^a,\qquad \hat{J^\nu_{\,\,a}}=q^\nu_\mu J^\mu_{\,\,a}.\nonumber
\end{eqnarray}
The lemma is proved.
\end{proof}

Combining gauge and orthogonal transformations, we conclude that the system (\ref{eq}) is invariant under the transformation
\begin{eqnarray}
&&A^\nu_b \to  \acute{\hat{A^\nu_{\,\,b}}}=q^\nu_\mu A^\mu_{\,\,a} p^a_b,\quad A_{n\times 3}\to \acute{\hat{A}}_{n \times 3}=Q_{n\times n}A_{n\times 3}P_{3\times3},\label{tr2}\\
&&J^\nu_b \to  \acute{\hat{J^\nu_{\,\,b}}}=q^\nu_\mu J^\mu_{\,\,a} p^a_b,\quad J_{n\times 3}\to \acute{\hat{J}}_{n \times 3}=Q_{n\times n}J_{n\times 3}P_{3\times3}\nonumber
\end{eqnarray}
for any $P\in\SO(3)$ and $Q\in\OO(n)$.

\begin{theorem}\label{th1} Let $A=||A^\nu_{\,\,k}||$, $J=||J^\nu_{\,\, k}||$ satisfy the system of $3n$ cubic equations (\ref{eq}). Then there exist matrices $P\in\SO(3)$ and $Q\in\OO(n)$ such that $QAP$ is diagonal. For all such matrices $P$ and $Q$, the matrix $QJP$ is diagonal too and the system (\ref{eq}) takes the following form under the transformation (\ref{tr2}):
\begin{eqnarray}
- a_1((a_2)^2+(a_3)^2)=j_1,\nonumber\\
- a_2((a_1)^2+(a_3)^2)=j_2,\label{s2}\\
- a_3((a_1)^2+(a_2)^2)=j_3\nonumber
\end{eqnarray}
in the case $n\geq 3$ and
\begin{eqnarray}
- a_1(a_2)^2=j_1,\label{s22}\\
- a_2(a_1)^2=j_2\nonumber
\end{eqnarray}
in the case $n=2$.

We denote diagonal elements of the matrix $QAP$ by $a_1$, $a_2$, $a_3$ (or $a_1$, $a_2$) and diagonal elements of the matrix $QJP$ by $j_1$, $j_2$, $j_3$ (or $j_1$, $j_2$).
\end{theorem}

\begin{proof} Let the system (\ref{eq}) has some solution $A^\mu_{\,\,a}$, $J^\mu_{\,\,a}$. Let us synchronize gauge transformation and orthogonal transformation such that $A=||A^\mu_{\,\, a}||$ will have a diagonal form. Namely, we take $P\in\SO(3)$ and $Q\in\OO(n)$ such that $QAP$ is diagonal. Note that we can always find the matrix $R\in\SO(N)$ from the special orthogonal group in SVD (\ref{hs1}). If it has the determinant $-1$, then we can change the sign of the first columns of the matrices $L$ and $R$ and the determinant will be $+1$.

Let us consider the case $n\geq 3$. In (\ref{eq}), we must take $\mu=a=c$, $\nu=b$ to obtain nonzero summands. Also we need $b=k$, i.e. $\nu=k$. In this case, the product of two Levi-Civita symbols in (\ref{eq}) equals $-1$. If $\nu\neq k$, then the expression on the left side of the equation equals zero. If $\nu=k$, then we obtain the following sum over index $\mu=a=c$
$$-\acute{\hat{A^k_{\,\,k}}}\sum_{a\neq k} (\acute{\hat{A^a_{\,\,a}}})^2,$$
where we have only 2 summands in the sum (except the value $\mu=k$ because of the Levi-Civita symbols).

Under our transformation the expressions $J^\mu_{\,\,a}$ are transformed into some new expressions $\acute{\hat{J^\mu_{\,\,a}}}$.
We obtain the following system of $3n$ equations
\begin{eqnarray}
- \acute{\hat{A^1_{\,\,1}}}((\acute{\hat{A^2_{\,\, 2}}})^2+(\acute{\hat{A^3_{\,\,3}}})^2)=\acute{\hat{J^1_{\,\,1}}},\nonumber\\
- \acute{\hat{A^2_{\,\,2}}}((\acute{\hat{A^1_{\,\,1}}})^2+(\acute{\hat{A^3_{\,\,3}}})^2)=\acute{\hat{J^2_{\,\,2}}},\label{Aeq}\\
- \acute{\hat{A^3_{\,\,3}}}((\acute{\hat{A^1_{\,\,1}}})^2+(\acute{\hat{A^2_{\,\,2}}})^2)=\acute{\hat{J^3_{\,\,3}}},\nonumber\\
0=\acute{\hat{J^\nu_{\,\,k}}},\qquad \nu\neq k,\qquad \nu=1, \ldots, n,\qquad k=1, 2, 3.\nonumber
\end{eqnarray}
This system of equations has solutions if the matrix $\acute{\hat{J}}$ is also diagonal.

In the case $n=2$, we obtain the system
\begin{eqnarray}
- \acute{\hat{A^1_{\,\,1}}}(\acute{\hat{A^2_{\,\, 2}}})^2=\acute{\hat{J^1_{\,\,1}}},\nonumber\\
- \acute{\hat{A^2_{\,\,2}}}(\acute{\hat{A^1_{\,\,1}}})^2=\acute{\hat{J^2_{\,\,2}}},\label{Aeq2}\\
0=\acute{\hat{J^\nu_{\,\,k}}},\qquad \nu\neq k,\qquad \nu=1, 2,\qquad k=1, 2, 3\nonumber
\end{eqnarray}
instead of the system (\ref{Aeq}) and the proof is similar.
\end{proof}

\textbf{Remark 1.} In Theorem \ref{th1}, we calculate SVD of the matrix $A$ and obtain non-negative singular values $a_1$, $a_2$, $a_3$ (or $a_1$, $a_2$ in the case $n=2$). The diagonal elements of the matrix $QJP$ will be non-positive because of the equations (\ref{s2}) (or (\ref{s22})). If we want, we can change the matrix $Q\in\OO(n)$ (multiplying by the matrix ${\bf -1}\in\OO(n)$) such that the elements of the new matrix $QJP$ will be non-negative (they will be singular values of the matrix $J$) and the elements of the new matrix $QAP$ will be non-positive. Multiplying the matrices $P$ and $Q$ by permutation matrices, which are also orthogonal, we can obtain the diagonal elements of the new matrix $QJP$ in decreasing order, the diagonal elements of the new matrix $QAP$ will be in some other order.

\textbf{Remark 2.} Suppose we have known matrix $J$ and want to obtain all solutions $A$ of the system (\ref{eq}). We can always calculate singular values $j_1$, $j_2$, $j_3$ (or $j_1$, $j_2$) of the matrix $J$ and solve the system (\ref{s2}) (or (\ref{s22})) using lemmas below. Finally, we obtain all solutions $A_D=\diag(a_1, a_2, a_3)$ (or $A_D=\diag(a_1, a_2)$) of the system (\ref{eq}) but in some other system of coordinates depending on $Q\in\OO(n)$ and with gauge fixing depending on $P\in\SO(3)$. The matrix
$$A=Q^{-1} A_D P^{-1}$$
will be solution of the system (\ref{eq}) in the original system of coordinates and with the original gauge fixing.

\textbf{Remark 3.} Note that $Q^{-1} Q_1^{-1} A_D P_1^{-1} P^{-1}$, for all $Q_1\in\OO(n)$ and $P_1\in\SO(3)$ such that $Q_1 J_D P_1=J_D$, will be also solutions of the system (\ref{eq}) in the original system of coordinates and with the original gauge fixing because of Lemma \ref{lemma1}. Here we denote $J_D=\diag(j_1, j_2, j_3)$ (or $J_D=\diag(j_1, j_2)$).

Let us give one example. If the matrix $J=0$, then all singular values of this matrix equal zero and we can take $Q=P={\bf 1}$ for its SVD. We solve the system (\ref{s2}) (or (\ref{s22})) for $j_1=j_2=j_3=0$ (or $j_1=j_2=0$) and obtain all solutions $A_D=\diag(a_1, a_2, a_3)$ (or $A_D=\diag(a_1, a_2)$) of this system. We have $Q_1 J_D P_1=J_D$ for $J_D=0$ and any $Q_1\in\OO(n)$, $P_1\in\SO(3)$. Therefore, the matrices $Q_1 A_D P_1$ for all $Q_1\in\OO(n)$ and $P_1\in\SO(3)$ will be solutions of the system (\ref{eq}) because of Lemma \ref{lemma1}.

Let us present a general solution of the systems (\ref{s2}) and (\ref{s22}) and discuss symmetries of these systems.


The systems (\ref{s2}), (\ref{s22}) can be rewritten in the following form using $b_k:=-a_k$, $k=1, 2, 3$ or $k=1, 2$:
\begin{eqnarray}
&&n\geq 3:\qquad b_1(b_2^2+b_3^2)=j_1,\quad b_2(b_1^2+b_3^2)=j_2,\quad b_3(b_1^2+b_2^2)=j_3,\label{eqn33}\\
&&n=2:\qquad b_1 b_2^2=j_1,\qquad b_2 b_1^2=j_2.\label{eqn233}
\end{eqnarray}
In the following lemmas, we assume that $j_1$, $j_2$, $j_3$ are known (parameters), and $b_1$, $b_2$, $b_3$ are unknown. We give  general solutions of the corresponding systems of equations.

The system (\ref{eqn33}) has the following symmetry. Suppose that $(b_1, b_2, b_3)$ is a solution of (\ref{eqn33}) for known $(j_1, j_2, j_3)$. If we change the sign of some $j_k$, $k=1, 2, 3$, then we must change the sign of the corresponding $b_k$, $k=1, 2, 3$. Thus, without loss of generality, we can assume that all expressions $b_k$, $j_k$, $k=1, 2, 3$ in (\ref{eqn33}) are non-negative. 
Similarly for the system (\ref{eqn233}).

\begin{lemma} \label{thn2} The system of equations (\ref{eqn233}) has the following general solution:
\begin{enumerate}
  \item in the case $j_1=j_2=0$, has solutions
$(b_1, 0)$, $(0, b_2)$ for all $b_1, b_2\in \R$;
  \item in the cases $j_1=0$, $j_2\neq 0$; $j_1\neq 0$, $j_2=0$, has no solutions;
  \item in the case $j_1\neq 0$, $j_2\neq 0$,
   has a unique solution
  $$b_1=\sqrt[3]{\frac{j_2^2}{j_1}},\qquad b_2=\sqrt[3]{\frac{j_1^2}{j_2}}.$$
\end{enumerate}
\end{lemma}
\begin{proof}
The proof is by direct calculation.
\end{proof}

The system (\ref{eqn33}) has the following symmetry.

\begin{lemma}\label{thSym}
If the system (\ref{eqn33}) has a solution $(b_1, b_2, b_3)$, where $b_1\neq 0$, $b_2\neq 0$, $b_3\neq 0$, then this system has also a solution $(\frac{K}{b_1},\frac{K}{b_2},\frac{K}{b_3})$, where $K=(b_1 b_2 b_3)^{\frac{2}{3}}.$
\end{lemma}

\begin{proof} Let us substitute $(\frac{K}{b_1},\frac{K}{b_2},\frac{K}{b_3})$ into the first equation. We have
$$j_1=4\frac{K}{b_1}(\frac{K^2}{b_2^2}+\frac{K^2}{b_3^2})=\frac{K^3(b_2^2+b_3^2)}{b_1b_2^2b_3^2}.$$
Using $j_1=b_1(b_2^2+b_3^2)$, we obtain
$$K=(b_1 b_2 b_3)^{\frac{2}{3}}.$$
We can verify that the same will be for the other two equations.
\end{proof}

For example, let us take $j_1=13$, $j_ 2=20$, $j_3=15$. Then the system (\ref{eqn33}) has solutions $(b_1, b_2, b_3)=(1, 2, 3)$ and $(6^{\frac{2}{3}},\frac{6^{\frac{2}{3}}}{2},\frac{6^{\frac{2}{3}}}{3})$.

\begin{lemma}\label{thGenSol} The system of equations (\ref{eqn33}) has the following general solution:
\begin{enumerate}
  \item in the case $j_1=j_2=j_3=0$, has solutions
$(b_1, 0, 0)$, $(0, b_2, 0)$, and $(0, 0, b_3)$ for all $b_1, b_2, b_3\in \R$;
  \item in the cases $j_1=j_2=0$, $j_ 3\neq 0$ (or similar cases with circular permutation), has no solutions;
  \item in the case $j_ 1\neq0$, $j_ 2\neq 0$, $j_ 3=0$ (or similar cases with circular permutation),
   has a unique solution
  $$b_1=\sqrt[3]{\frac{j_2^2}{j_1}},\qquad b_ 2=\sqrt[3]{\frac{j_1^2}{j_ 2}},\qquad b_ 3=0;$$

 \item in the case $j_1=j_2=j_3\neq 0$, has a unique solution
 $$b_ 1=b_ 2=b_ 3=\sqrt[3]{\frac{j_1}{2}};$$

 \item in the case of not all the same $j_1, j_2, j_3>0$ (and we take positive for simplicity), has the following two solutions
     $$(b_ {1+}, b_ {2+}, b_ {3+}),\qquad (b_ {1-}, b_ {2-}, b_ {3-})$$
     with the following expression for $K$ from Lemma \ref{thSym}
$$K:=b_{1+}b_{1-}=b_{2+}b_{2-}=b_{3+}b_{3-}=(b_{1+}b_{2+}b_{3+})^{\frac{2}{3}}=(b_{1-}b_{2-}b_{3-})^{\frac{2}{3}}:$$
       \begin{enumerate}
        \item in the case $j_1=j_2>j_3>0$ (or similar cases with circular permutation):
\begin{eqnarray}
&&b_{1\pm}=b_{2\pm}=\sqrt[3]{\frac{j_3}{2z_{\pm}}},\quad b_{3\pm}=z_{\pm} b_{1\pm},\quad z_{\pm}=\frac{j_1\pm\sqrt{j_1^2-j_3^2}}{j_3}.\nonumber\\
&&\mbox{Moreover,}\quad z_+ z_-=1,\quad K=(\frac{j_3}{2})^\frac{2}{3}.\nonumber
\end{eqnarray}
        \item in the case $j_3>j_1=j_2>0$ (or similar cases with circular permutation):
\begin{eqnarray}
&&b_{1\pm}=\frac{1}{w_\pm}b_{3},\quad b_{2\pm}=w_{\pm}b_{3},\quad b_{3\pm}=b_{3}=\sqrt[3]{\frac{j_1}{s}},\nonumber\\
&&w_{\pm}=\frac{s\pm\sqrt{s^2-4}}{2},\quad s=\frac{j_3+\sqrt{j_3^2+8j_1^2}}{2j_1}.\nonumber\\
&&\mbox{Moreover,}\quad w_+ w_-=1,\quad b_{1\pm}=b_{2\mp},\quad K=(\frac{j_1}{s})^\frac{2}{3}.\nonumber
\end{eqnarray}
        \item in the case of all different $j_1, j_2, j_3>0$:
\begin{eqnarray}
&&b_{1\pm}=\sqrt[3]{\frac{j_3}{t_0 y_\pm z_\pm}},\quad b_{2\pm}=y_{\pm} b_{1\pm},\quad b_{3\pm}=z_{\pm} b_{1\pm},\nonumber\\
&&z_{\pm}=\sqrt{\frac{y_\pm(j_1-j_2 y_\pm)}{j_ 2-j_1 y_\pm}},\quad y_{\pm}=\frac{t_0\pm\sqrt{t_0^2-4}}{2},\nonumber
\end{eqnarray}
where $t_0>2$ is the solution (it always exists, moreover, it is bigger than $\frac{j_2}{j_1}+\frac{j_ 1}{j_2}$) of the cubic equation $$j_1 j_2 t^3-(j_1^2+j_2^2+j_3^2)t^2+4j_3^2=0.$$
\begin{eqnarray}
\mbox{Moreover,}\quad y_+ y_-=1,\quad z_+ z_-=1,\quad K=(\frac{j_3}{t_0})^\frac{2}{3}.\nonumber
\end{eqnarray}
We can use the explicit Vieta or Cardano formulas for $t_0$:
$$t_0=\Omega+2 \Omega\cos(\frac{1}{3}\arccos(1-\frac{2\beta}{\Omega^3})),$$
$$\Omega:=\frac{\alpha+\beta}{3},\qquad \alpha:=A+\frac{1}{A}>2,\qquad \beta:=\frac{B^2}{A},\qquad A:=\frac{j_2}{j_1},\qquad B:=\frac{j_3}{j_1},$$
$$t_0=\Omega+L+ \frac{\Omega^2}{L},\qquad L:=\sqrt[3]{\Omega^3-2\beta+2\sqrt{\beta(\beta-\Omega^3)}}.$$
\end{enumerate}
\end{enumerate}
\end{lemma}
\begin{proof} The proof is rather cumbersome, we give it in Appendix A.
\end{proof}

\section{Results for the potential and the strength of the Yang-Mills field}\label{secstr}

In the case of the constant potential of the Yang-Mills field, we have the following expression for the strength
\begin{eqnarray}
F^{\mu\nu}=- [A^\mu, A^\nu]=-  [A^\mu_{\,\,a} \tau^a, A^\nu_{\,\,b} \tau^b]=-  A^\mu_{\,\,a} A^\nu_{\,\,b} \epsilon^{ab}_{\,\,\,\,\, c}\tau^c=F^{\mu\nu}_{\,\,\,\,\,c} \tau^c.\label{F}
\end{eqnarray}
If we take a system of coordinates depending on $Q\in\OO(n)$ and a gauge fixing depending on $P\in\SO(3)$ such that the matrices $A=||A^\mu_{\,\, a}||$ and $J=||J^\mu_{\,\, a}||$ are diagonal (see Theorem \ref{th1}), then the expressions $F^{\mu\nu}_{\,\,\,\,\,c}$
are nonzero only in the case of three different indices $\mu=a$, $\nu=b$, and $c$, which take the values $1, 2, 3$. For each solution, we calculate the invariant $F^2=F_{\mu\nu}F^{\mu\nu}$, which is present in the Lagrangian of the Yang-Mills field.

Using results of the previous section for the system (\ref{eq}) and the expression (\ref{F}), we obtain the following results for the potential $A$ and strength $F$ of the Yang-Mills field depending on the constant current $J$. The case $n=2$ is much simpler than the case $n\geq 3$, we discuss this case for the sake of completeness.

\textbf{In the case of dimension $n=2$:}
\begin{enumerate}
  \item In the case of zero current $J=0$, we have zero potential $A=0$ or nonzero potential (see Case 1 of Lemma \ref{thn2})
\begin{eqnarray}
A=\left(
     \begin{array}{ccc}
       a & 0 & 0 \\
       0 & 0 & 0 \\
     \end{array}
   \right),\qquad a\in \R\setminus\{0\}.\nonumber
\end{eqnarray}
In these cases, we have zero strength $F=0$ ($F^{\mu\nu}=0$). Note that this fact is already known (see \cite{Sch}, \cite{Sch2}).
 \item In the case $\rank (J)=1$, we have no constant solutions (see Case 2 of Lemma \ref{thn2}).
\item In the case $\rank (J)=2$, we have a unique solution (see Case 3 of Lemma \ref{thn2})
$$A=\left(
     \begin{array}{ccc}
       a_1 & 0 & 0 \\
       0 & a_2 & 0 \\
     \end{array}
   \right),\qquad a_1=-\sqrt[3]{\frac{j_2^2}{j_1}},\qquad a_2=\sqrt[3]{\frac{j_1^2}{j_2}}.
$$
For the strength, we have the following nonzero components
\begin{eqnarray}
F^{12}=-F^{21}=-\sqrt[3]{j_1 j_2}\tau^3\label{F1}
\end{eqnarray}
using specific system of coordinates and specific gauge fixing, where $j_1$ and $j_2$ are singular values of the matrix $J=||J^\mu_{\,\, a}||$. In this case, we obtain the following expression for the invariant $F^2=F_{\mu\nu}F^{\mu\nu}$:
\begin{eqnarray}
F^2=F_{\mu\nu}F^{\mu\nu}=-\frac{1}{2}\sqrt[3]{(j_1 j_2)^2} {\bf 1}\neq 0.\label{F2}
\end{eqnarray}
\end{enumerate}

\newpage

\textbf{In the case of dimension $n\geq 3$:}
  \begin{enumerate}
    \item In the case $J=0$, we have zero potential $A=0$ or nonzero potential (see Case 1 of Lemma \ref{thGenSol}):
\begin{eqnarray}
A=\left(
     \begin{array}{ccc}
       a & 0 & 0 \\
       0 & 0 & 0 \\
       \ldots & \ldots & \ldots \\
       0 & 0 & 0\\
     \end{array}
   \right),\qquad a\in \R\setminus\{0\}.\label{A}
\end{eqnarray}
In these cases, we have zero strength $F=0$.
    \item In the case $\rank (J)=1$, we have no constant solutions (see Case 2 of Lemma \ref{thGenSol}).
    \item In the case $\rank (J)=2$, we have a unique solution (see Case 3 of Lemma \ref{thGenSol}):
\begin{eqnarray}
A=\left(
     \begin{array}{ccc}
       a_1 & 0 & 0 \\
       0 & a_2 & 0 \\
       0 & 0 & 0\\
       \ldots & \ldots & \ldots \\
       0 & 0 & 0\\
     \end{array}
   \right),\qquad a_1=-\sqrt[3]{\frac{j_2^2}{j_1}},\qquad a_2=\sqrt[3]{\frac{j_1^2}{j_2}}.\label{A5}
\end{eqnarray}
For the strength, we have the following nonzero components (\ref{F1}) and again (\ref{F2}) using specific system of coordinates and gauge fixing, where $j_1$, $j_2$, and $j_3=0$ are singular values of the matrix $J$.
    \item In the case $\rank (J)=3$, we have one or two solutions.

In the specific case of all the same singular values $j:=j_1=j_2=j_3\neq 0$, we have a unique solution (see Case 4 of Lemma \ref{thGenSol})
\begin{eqnarray}
A=\left(
     \begin{array}{ccc}
       a & 0 & 0 \\
       0 & a & 0 \\
       0 & 0 & a\\
       0 & 0 & 0 \\
       \ldots & \ldots & \ldots \\
       0 & 0 & 0\\
     \end{array}
   \right),\qquad a=-\sqrt[3]{\frac{j}{2}}.\label{A6}
\end{eqnarray}
We have the following nonzero components of the strength:
\begin{eqnarray}
&&F^{12}=-F^{21}=-\sqrt[3]{\frac{j^2}{4}}\tau^3, \quad F^{23}=-F^{32}=-\sqrt[3]{\frac{j^2}{4}}\tau^1, \nonumber\\
&&F^{31}=-F^{13}=-\sqrt[3]{\frac{j^2}{4}}\tau^2.\label{F6}
\end{eqnarray}
In this case, we have
\begin{eqnarray}
F^2=F_{\mu\nu}F^{\mu\nu}=\frac{-3}{2}\sqrt[3]{\frac{j^4}{16}}{\bf 1}\neq 0.\label{F26}
\end{eqnarray}
In the case of not all the same singular values $j_1$, $j_2$, $j_3$ of the matrix $J$, we have two different solutions
\begin{eqnarray}
A=\left(
     \begin{array}{ccc}
       -b_{1\pm} & 0 & 0 \\
       0 & -b_{2\pm} & 0 \\
       0 & 0 & -b_{3\pm}\\
       0 & 0 & 0 \\
       \ldots & \ldots & \ldots \\
       0 & 0 & 0\\
     \end{array}
   \right),\label{A7}
\end{eqnarray}
where $b_{k\pm}$, $k=1, 2, 3$ are from Case 5 of Lemma \ref{thGenSol}.

We have the following nonzero components of the strength:
\begin{eqnarray}
&&F_{\pm}^{12}=-F_{\pm}^{21}=-b_{1\pm} b_{2\pm} \tau^3,\quad F_{\pm}^{23}=-F_{\pm}^{32}=-b_{2\pm} b_{3\pm} \tau^1,\nonumber\\
&&F_{\pm}^{31}=-F_{\pm}^{13}=-b_{3\pm} b_{1\pm}\tau^2.\label{F7}
\end{eqnarray}
In this case, we have
\begin{eqnarray}
F_{\pm}^2=F_{\mu\nu\,\pm} F_{\pm}^{\mu\nu}=-\frac{1}{2}((b_{1\pm}b_{2\pm})^2+(b_{2\pm}b_{3\pm})^2+(b_{3\pm}b_{1\pm})^2){\bf 1}\neq 0.\label{F100}
\end{eqnarray}
\end{enumerate}
In the next lemma, we give the explicit form of (\ref{F100}).

\begin{lemma}\label{lemF2} In the case of not all the same $j_1$, $j_2$, $j_3$, (\ref{F100}) takes the form:
\begin{enumerate}
  \item in the case $j_1=j_2>j_3>0$ (or similar cases with circular permutation):
\begin{eqnarray}
F^2_{\pm}=\frac{-K^2 (1+2z_{\pm}^2)}{2 z_{\pm}^{\frac{4}{3}}}{\bf 1},\qquad F^2_+\neq F^2_-,\label{F21}\\
\mbox{where}\quad z_{\pm}=\frac{j_1\pm\sqrt{j_1^2-j_3^2}}{j_3},\quad K=(\frac{j_3}{2})^{\frac{2}{3}}.\nonumber
\end{eqnarray}
  \item in the case $j_3>j_1=j_2>0$ (or similar cases with circular permutation):
\begin{eqnarray}
F^2_{\pm}=\frac{-K^2(s^2-1)}{2}{\bf 1},\qquad  F^2_+=F^2_-,\label{F22}\\
\mbox{where}\quad s=\frac{j_3+\sqrt{j_3^2+8j_1^2}}{2j_1}>2,\quad K=(\frac{j_1}{s})^{\frac{2}{3}}.\nonumber
\end{eqnarray}
  \item in the case of all different $j_1$, $j_2$, $j_3 >0$:
\begin{eqnarray}
F^2_{\pm}=\frac{-K^2 (y_{\pm}^2+z_{\pm}^2+y_{\pm}^2 z_{\pm}^2)}{2(y_{\pm} z_{\pm})^{\frac{4}{3}}}{\bf 1},\qquad  F^2_+\neq F^2_-,\label{F23}\\
\mbox{where $K=(\frac{j_3}{t_0})^{\frac{2}{3}}$, and $y_{\pm}$, $z_\pm$, $t_0$ are from Case (5) - (c) of Lemma \ref{thGenSol}.}\nonumber
\end{eqnarray}
\end{enumerate}
In all cases of Lemma, the expression $K$ is the invariant for each pair of solutions (see Lemmas \ref{thSym} and \ref{thGenSol}).
\end{lemma}
\begin{proof} We give the proof in Appendix B.
\end{proof}

Note that in Case 2 of Lemma \ref{lemF2}, we have two constant solutions of the Yang-Mills equations with the same invariant $F^2=F^2_\pm$. In each of two Cases 1 and 3, we have two constant solutions of the Yang-Mills equations with different invariants $F_+^2\neq F^2_-$.

We summarize the results for the case of arbitrary Euclidean space $\R^n$, $n\geq 2$, in Table 1.

\begin{table}[ht]
\tbl{All constant solutions of $\SU(2)$ Yang-Mills equations in $\R^n$.}{%
\begin{tabular}{@{\quad}c@{\qquad}c@{\qquad}c@{\qquad}c@{\qquad}c@{\qquad}c@{\qquad}c@{\qquad}} \toprule
$n$ & $\rank (J)$ & additional conditions & $\rank (A)$ & $A$ & $F$ & $F^2$\\ \colrule
$n\geq 2$ & 0 &  & 0 & $A=0$ & $F=0$ & $F^2=0$  \\
  $n\geq 2$ & 0 &  & 1 & see (\ref{A}) & $F=0$ & $F^2=0$  \\
  $n\geq 2$ & 1 &  &   & $\o$ & $\o$ & $\o$  \\
  $n\geq 2$ & 2 &  & 2 & see (\ref{A5})& see (\ref{F1}) & see (\ref{F2})  \\
  $n\geq 3$ & 3 & $j_1=j_2=j_3$ & 3 & see (\ref{A6}) & see (\ref{F6}) & see (\ref{F26})  \\
  $n\geq 3$ & 3 & $j_1=j_2>j_3$ & 3 & see (\ref{A7}) & see (\ref{F7}) & see (\ref{F21})   \\
  $n\geq 3$ & 3 & $j_3>j_1=j_2$ & 3 & see (\ref{A7}) & see (\ref{F7}) & see (\ref{F22})   \\
  $n\geq 3$ & 3 & all different $j_1, j_2, j_3$ & 3 & see (\ref{A7}) & see (\ref{F7}) & see (\ref{F23})   \\ \botrule
\end{tabular}}
\end{table}

\section{Conclusions}

The main result of this paper is the presentation of all constant solutions of the Yang-Mills equations with $\SU(2)$ gauge symmetry for an arbitrary constant current in Euclidean space of arbitrary finite dimension. Using the invariance of the Yang-Mills equations under the orthogonal transformations of coordinates and gauge invariance, we choose a specific system of coordinates and a specific gauge fixing for each constant current and obtain all constant solutions of the Yang-Mills equations in this system of coordinates with this gauge fixing, and then in the original system of coordinates with the original gauge fixing (see Remarks 2 and 3). We prove that the number ($0$, $1$, or $2$) of constant solutions of the Yang-Mills equations (solutions of the system (\ref{eq})) in terms of the strength $F$ (\ref{F}) depends on the rank of the matrix $J$ and, sometimes, on the singular values of this matrix (see Section \ref{secstr}). The explicit form of these solutions and the invariant $F^2$ can always be written using singular values of the matrix $J$.

We plan to solve the same problem as in this paper, but in pseudo-Euclidean space of arbitrary finite dimension, in particular, in the case of Minkowski space. This will allow us to obtain all constant solutions of the Dirac-Yang-Mills equations, which is interesting for applications. Another task is to consider the same problem on curved manifolds. We need another technique to solve the same problem for the case of the Lie group $SU(3)$, which is important for describing strong interactions.

Note that now we can consider nonconstant solutions of Yang-Mills equations in the form of series of perturbation theory using all constant solutions from Lemmas \ref{thn2} and \ref{thGenSol} as a zeroth approximation. The problem reduces to solving systems of linear partial differential equations. This will allow us to give a local classification of all solutions of the classical $\SU(2)$ Yang-Mills equations.

The results of this paper are new and can be used to solve some problems in particle physics, in particular, in describing physical vacuum \cite{Actor}, \cite{Gr}, \cite{Ja}, \cite{nian}. In this paper, we discuss mathematical structures and constructions. Relating the proposed mathematical constructions to real world objects goes beyond the scope of this investigation.  The explicit formulas for solutions (see the results of Section 3 and Lemmas \ref{thn2}, \ref{thSym}, and \ref{thGenSol}) are fundamental for the Yang-Mills field and should be interesting for physicists.

\section*{Acknowledgements}

The author is grateful to N. G. Marchuk for fruitful discussions.
The author is grateful to the participants of the 2nd JNMP Conference on Nonlinear Mathematical Physics (Chile, Santiago, 2019) for useful comments. The author is grateful to the reviewers for their careful reading of the paper and helpful remarks.

This work is supported by the Russian Science Foundation (project 18-71-00010).

\appendix

\section{The proof of Lemma \ref{thGenSol}.}

The first four cases of Lemma \ref{thGenSol} are easily verified.

Let us consider the case of not all the same positive $j_1, j_2, j_3$. As we mentioned before the lemma, we can assume that $j_k>0$ and $b_k>0$ because if we change the sign of $j_k$, then the sign of $b_k$ is also changed.

We use the following change of variables
$$x=b_1>0,\qquad y=\frac{b_2}{b_1}>0,\qquad z=\frac{b_3}{b_1}>0.$$
We obtain
\begin{eqnarray}
j_1=x^3(y^2+z^2),\quad
j_2=y x^3(1+z^2),\quad
j_3=z x^3(1+y^2).
\end{eqnarray}
Using notation
$$A=\frac{j_2}{j_1}>0,\qquad B=\frac{j_3}{j_1}>0,$$
we get the system for two variables $y$ and $z$:
\begin{eqnarray}y(1+z^2)=A(y^2+z^2),\qquad z(1+y^2)=B(y^2+z^2).\label{2ur}\end{eqnarray}
For the variable $x=b_1$, we have
\begin{eqnarray}
b_1=x=\sqrt[3]{\frac{j_1}{(y^2+z^2)}}.\label{x}
\end{eqnarray}
From the first equation (\ref{2ur}), we obtain
\begin{eqnarray}
(A-y)z^2=y(1-Ay).\label{67}
\end{eqnarray}
Let us consider two cases: $A=1$ and $A\neq 1$.

If $A=1$ (and $B\neq 1$, because we consider the case of not all the same $j_1$, $j_2$, $j_3$), then we can rewrite $(1-y)z^2=y(1-y)$ in the form $(1-y)(y-z^2)=0$. If $y=1$, then we substitute this condition into the second equation (\ref{2ur}) and obtain $2z=B(1+z^2)$. If $0<B<1$, then $z=\frac{1\pm\sqrt{1-B^2}}{B}$. For $B>1$, there is no solution of this type. If $y=z^2$, then we substitute this condition into the second equation (\ref{2ur}) and get $z(1+z^4)=B(z^2+z^4)$, $z^5-Bz^4-Bz^2+z=0$. Dividing both sides by $z^2$ and using notation $s=z+\frac{1}{z}>0$, we get $s=\frac{B\pm\sqrt{B^2+8}}{2}$. We have $z>0$ and $s>2$ (if $s=2$, then $z=y=1$ and we obtain the case $j_ 1=j_2=j_3$, which is not considered now). Therefore, $B>1$ and we have two expressions $z=\frac{s\pm\sqrt{s^2-4}}{2}$ for one $s=\frac{B+\sqrt{B^2+8}}{2}$.

The results for the case $A=1$ can be summarized as follows. If $j_1=j_2>j_3$ ($A=1$, $0<B<1$), then we have two solutions
$$b_{1\pm}=b_{2\pm}=\sqrt[3]{\frac{j_1}{1+z_{\pm}^2}},\qquad b_{3\pm}=z_{\pm} b_{1\pm},\qquad z_{\pm}=\frac{1\pm\sqrt{1-B^2}}{B},\qquad B=\frac{j_3}{j_1}.$$
Using $z_+ z_-=1$ and $\frac{z_\pm}{1+z_\pm^2}=\frac{B}{2}$, we can verify that $K=b_{1+}b_{1-}=b_{2+}b_{2-}=b_{3+}b_{3-}= (b_{1+}b_{2+}b_{3+})^{\frac{2}{3}}=(b_{1-}b_{2-}b_{3-})^{\frac{2}{3}}= (\frac{Bj_1}{2})^\frac{2}{3}=(\frac{j_3}{2})^\frac{2}{3}$
and rewrite solutions in the following form
$$b_{1\pm}=b_{2\pm}=\sqrt[3]{\frac{j_3}{2z_{\pm}}},\qquad b_{3\pm}=z_{\pm} b_{1\pm},\qquad z_{\pm}=\frac{j_1\pm\sqrt{j_1^2-j_3^2}}{j_3}.$$
If $j_3>j_1=j_2$ ($A=1$, $B>1$), then we have two solutions
$$b_{1\pm}=\sqrt[3]{\frac{j_1}{w^2+w^4}},\qquad b_{2\pm}=w_{\pm}^2b_{1\pm},\qquad b_{3\pm}=w_{\pm} b_{1\pm},$$ $$w_{\pm}=\frac{s\pm\sqrt{s^2-4}}{2},\qquad s=\frac{B+\sqrt{B^2+8}}{2},\qquad B=\frac{j_3}{j_1}.$$
Using $\frac{w_{\pm}}{1+w_{\pm}^2}=\frac{1}{s}$ and $w_+ w_-=1$, we can verify that $K=b_{1+}b_{1-}=b_{2+}b_{2-}=b_{3+}b_{3-}= (b_{1+}b_{2+}b_{3+})^{\frac{2}{3}}=(b_{1-}b_{2-}b_{3-})^{\frac{2}{3}}= (\frac{j_1}{s})^\frac{2}{3}$. Moreover, $b_{3+}=b_{3-}$, $b_{1+}=b_{2-}$, $b_{1-}=b_{2+}$ and we can rewrite solutions in the following form
$$b_{1\pm}=\frac{1}{w_\pm}b_{3},\qquad b_{2\pm}=w_{\pm}b_{3},\qquad b_{3\pm}=b_{3}=\sqrt[3]{\frac{j_1}{s}},$$ $$w_{\pm}=\frac{s\pm\sqrt{s^2-4}}{2},\qquad s=\frac{j_3+\sqrt{j_3^2+8j_1^2}}{2j_1}.$$

Now let us consider the case $A\neq 1$. If $B=1$, then we can similarly consider this case as previous case ($A=1$, $B\neq 1$), since the system (\ref{2ur}) is symmetric w.r.t. the change $y \leftrightarrow z$, $A \leftrightarrow B$. Therefore, let us consider the remaining case, when $A\neq 1$ and $B\neq 1$. If $A=B$, then $j_2=j_3$ and this case is similar to the previous one again. Let us consider the case $A\neq B$, i.e. all $j_k$, $k=1,2,3$ are different now.

If $A\neq 1$, then $A\neq y$. Really, suppose that $A=y$. Then from (\ref{67}) we obtain $A=1$ and a contradiction. Since $A\neq y$, we get from (\ref{67})
\begin{eqnarray}
z^2=\frac{y(1-Ay)}{A-y}.\label{y2}
\end{eqnarray}
Because all variables are positive, the second equation (\ref{2ur}) is equivalent to $z^2(1+y^2)^2=B^2(z^2+y^2)^2$. Substituting (\ref{y2}) into this expression, we get
$$\frac{y(1-Ay)}{A-y}(1+y^2)^2=B^2(\frac{y(1-Ay)}{A-y}+y^2)^2.$$
Note, that when we will find $y$ from this equation, expression $\frac{y(1-Ay)}{A-y}$ will be positive. Therefore $z=\sqrt{\frac{y(1-Ay)}{A-y}}$. We have
$$\frac{y(1-Ay)(1+y^2)^2}{A-y}=\frac{B^2(y-y^3)^2}{(A-y)^2},\quad (1-Ay)(1+y^2)^2(A-y)=B^2y(1-y^2)^2,$$
$$Ay^6-(B^2+A^2+1)y^5+3Ay^4+2(B^2-A^2-1)y^3+3Ay^2-(B^2+A^2+1)y+A=0.$$
Dividing both sides by $y^3$ and using notation $t=y+\frac{1}{y}>0$ ($t^2=y^2+\frac{1}{y^2}+2$, $t^3=y^3+3y+3\frac{1}{y}+\frac{1}{y^3}$), we get
\begin{eqnarray}
A(y^3+\frac{1}{y^3})-(B^2+A^2+1)(y^2+\frac{1}{y^2})+3A(y+\frac{1}{y})+2(B^2-A^2-1)=0,\nonumber\\
A(t^3-3t)-(B^2+A^2+1)(t^2-2)+3At+2(B^2-A^2-1)=0,\nonumber\\
At^3-(B^2+A^2+1)t^2+4B^2=0.\nonumber
\end{eqnarray}
We obtain the following cubic equation
\begin{eqnarray}
f(t):=t^3-(\frac{B^2}{A}+A+\frac{1}{A})t^2+4\frac{B^2}{A}=0.\label{kub}
\end{eqnarray}
We are interested only in the positive solutions $t>0$ of this equation. We have $y^2-ty+1=0$. This equation has positive solutions $y=\frac{t\pm\sqrt{t^2-4}}{2}$ only in the case $t\geq 2$.

We have $f(-\infty)=-\infty$, $f(0)=\frac{4B^2}{A}>0$, $f(2)=-\frac{4(A-1)^2}{A}<0$, $f(+\infty)=+\infty$. This means that the cubic equation (\ref{kub}) has one negative solution, one solution between $0$ and $2$, and one solution $t_0>2$. Moreover, solution $t_0$ is bigger than $A+\frac{1}{A}$, because $f(A+\frac{1}{A})=-\frac{B^2 (A^2-1)^2}{A^3}<0$. We can calculate $t_0$ using Cardano formulas.

Finally, if $A\neq 1$ and $B\neq 1$, then we have two solutions
$$b_{1\pm}=\sqrt[3]{\frac{j_1}{y_{\pm}^2+z_{\pm}^2}},\qquad b_{2\pm}=y_{\pm} b_{1\pm},\qquad b_{3\pm}=z_{\pm} b_{1\pm},$$
$$z_{\pm}=\sqrt{\frac{y_\pm(1-Ay_\pm)}{A-y_\pm}},\qquad y_{\pm}=\frac{t_0\pm\sqrt{t_0^2-4}}{2},$$
where $t_0=t_0(A,B)>2$ is a solution of the equation (\ref{kub}).
Using $\frac{y_{\pm}z_{\pm}}{y_{\pm}^2+z_{\pm}^2}=\frac{B y_\pm}{1+y_\pm^2}=\frac{B}{t_0}$, $y_+ y_-=1$, and $z_+ z_-=1$, we can verify that $K=b_{1+}b_{1-}=b_{2+}b_{2-}=b_{3+}b_{3-}= (b_{1+}b_{2+}b_{3+})^{\frac{2}{3}}=(b_{1-}b_{2-}b_{3-})^{\frac{2}{3}}= (\frac{j_3}{t_0})^\frac{2}{3}$ and rewrite solutions in the following form
$$b_{1\pm}=\sqrt[3]{\frac{j_3}{t_0 y_\pm z_\pm}},\qquad b_{2\pm}=y_{\pm} b_{1\pm},\qquad b_{3\pm}=z_{\pm} b_{1\pm},$$ $$z_{\pm}=\sqrt{\frac{y_\pm(1-Ay_\pm)}{A-y_\pm}},\qquad y_{\pm}=\frac{t_0\pm\sqrt{t_0^2-4}}{2},$$
where $t_0=t_0(A, B)>2$ is a solution of the equation (\ref{kub}). The lemma is proved.

\section{The proof of Lemma \ref{lemF2}.}

1) In Case (5) - (a) of Lemma \ref{thGenSol}, using
$$b_{1\pm}=b_{2\pm},\qquad b_{3\pm}=z_{\pm}b_{1\pm},\qquad b_{1\pm}=\sqrt[3]{\frac{j_3}{2z_{\pm}}},\qquad K=(\frac{j_3}{2})^\frac{2}{3},$$
we obtain
\begin{eqnarray}
F_{\pm}^2&=&-\frac{1}{2}((b_{1\pm}b_{2\pm})^2+(b_{2\pm}b_{3\pm})^2+(b_{3\pm}b_{1\pm})^2){\bf 1}= -\frac{1}{2}b_{1\pm}^4(1+2z_{\pm}^2){\bf 1}\nonumber\\
&=&-\frac{1}{2}(\frac{j_3}{2z_{\pm}})^{\frac{4}{3}}(1+2z_{\pm}^2){\bf 1}=\frac{-K^2 (1+2z_{\pm}^2)}{2 z_{\pm}^{\frac{4}{3}}}{\bf 1}.\nonumber
\end{eqnarray}
Let us prove that $F^2_+ \neq F^2_-$ in this case. Suppose that we have $F^2_+=F^2_-$, i.e.
$$
\frac{1+2z_+^2}{z_+^{\frac{4}{3}}}=\frac{1+2z_-^2}{z_-^{\frac{4}{3}}},\qquad
z_-^{\frac{4}{3}}+2z_+^2z_-^{\frac{4}{3}}=z_+^{\frac{4}{3}}+2z_-^2z_+^{\frac{4}{3}}.
$$
Using $z_+ z_-=1$, we get
$$
z_-^{\frac{4}{3}}+2z_+^{\frac{2}{3}}=z_+^{\frac{4}{3}}+2z_-^{\frac{2}{3}},\qquad (z_-^{\frac{2}{3}}-1)^2=(z_+^{\frac{2}{3}}-1)^2,\qquad
(z_-^{\frac{2}{3}}+z_+^{\frac{2}{3}}-2)(z_-^{\frac{2}{3}}-z_+^{\frac{2}{3}})=0,
$$
which is not possible, because $z_-^{\frac{2}{3}}+z_+^{\frac{2}{3}}\geq 2 (z_- z_+)^{\frac{1}{3}}=2$ and $z_+$, $z_-$ do not equal $\pm 1$.

2) In Case (5) - (b) of Lemma \ref{thGenSol}, using
$$b_{1\pm}=\frac{1}{w_{\pm}}b_3,\qquad b_{2\pm}=w_{\pm}b_{3},\qquad b_{3\pm}=b_{3}=\sqrt[3]{\frac{j_1}{s}},\qquad K=(\frac{j_1}{s})^{\frac{2}{3}},$$
we obtain
\begin{eqnarray}
F^2_{\pm}&=&-\frac{1}{2}((b_{1\pm}b_{2\pm})^2+(b_{2\pm}b_{3\pm})^2+(b_{3\pm}b_{1\pm})^2){\bf 1}= -\frac{1}{2}b_{3\pm}^4(1+w_{\pm}^2+\frac{1}{w_{\pm}^2}){\bf 1}\nonumber\\
&=&-\frac{1}{2}(\frac{j_1}{s})^{\frac{4}{3}}(1+w_{\pm}^2+\frac{1}{w_{\pm}^2}){\bf 1}=-\frac{K^2}{2} (1+w_{\pm}^2+\frac{1}{w_{\pm}^2}){\bf 1}=-\frac{K^2}{2} (s^2-1){\bf 1}.\nonumber
\end{eqnarray}
In the last equality, we used $w_{\pm}+\frac{1}{w_{\pm}}=s$, i.e. $w_{\pm}^2+\frac{1}{w_{\pm}^2}=s^2-2$.

We have $F^2_+=F^2_-$, because $F_{\pm}^2$ does not depend on $w_{\pm}$ in this case.

3) In Case (5) - (c) of Lemma \ref{thGenSol}, using
$$b_{2\pm}=y_{\pm}b_{1\pm},\qquad b_{3\pm}=z_{\pm}b_{1\pm},\qquad b_{1\pm}=\sqrt[3]{\frac{j_3}{t_0 y_{\pm}z_{\pm}}},\qquad K=(\frac{j_3}{t_0})^{\frac{2}{3}},$$
we obtain
\begin{eqnarray}
F^2_{\pm}&=&-\frac{1}{2}((b_{1\pm}b_{2\pm})^2+(b_{2\pm}b_{3\pm})^2+(b_{3\pm}b_{1\pm})^2){\bf 1}= -\frac{1}{2}b_{1\pm}^4(y_{\pm}^2+z_{\pm}^2 + y_{\pm}^2 z_{\pm}^2){\bf 1}\nonumber\\
&=&-\frac{1}{2}(\frac{j_3}{t_0 y_{\pm}z_{\pm}})^{\frac{4}{3}}(y_{\pm}^2+z_{\pm}^2 + y_{\pm}^2 z_{\pm}^2){\bf 1}=-\frac{K^2(y_{\pm}^2+z_{\pm}^2 + y_{\pm}^2 z_{\pm}^2)}{2(y_{\pm}z_{\pm})^{\frac{4}{3}}} {\bf 1}.\nonumber
\end{eqnarray}
Using $z_\pm^2=\frac{y_{\pm}(1-Ay_{\pm})}{A-y_{\pm}}$, we also get
\begin{eqnarray}
F^2_{\pm}&=&-\frac{K^2(y_{\pm}^2+z_{\pm}^2 + y_{\pm}^2 z_{\pm}^2)}{2(y_{\pm}z_{\pm})^{\frac{4}{3}}} {\bf 1}=
-\frac{K^2(y_{\pm}^2+(1+ y_{\pm}^2)\frac{y_{\pm}(1-Ay_{\pm})}{A-y_{\pm}}))}{2y_{\pm}^{\frac{4}{3}}(\frac{y_{\pm}(1-Ay_{\pm})}{A-y_{\pm}})^{\frac{2}{3}}} {\bf 1}\nonumber\\
&=&\frac{-K^2(1-Ay_{\pm}^3)}{2y_{\pm}(A-y_{\pm})^{\frac{1}{3}}(1-Ay_{\pm})^{\frac{2}{3}}}{\bf 1}.\nonumber
\end{eqnarray}

Let us prove that $F^2_+ \neq F^2_-$ in this case. Suppose that we have $F^2_+=F^2_-$, i.e.
$$
\frac{(1-Ay_{+}^3)^3}{y^3_{+}(A-y_{+})(1-Ay_{+})^2}=\frac{(1-Ay_{-}^3)^3}{y^3_{-}(A-y_{-})(1-Ay_{-})^2}.
$$
Using $y_-=y_+^{-1}$, we get
$$y_+^3(A-y_+)(1-Ay_+)^2(1-\frac{A}{y_+^3})^3=(1-Ay_+^3)^3\frac{1}{y_+^3}(A-\frac{1}{y_+})(1-\frac{A}{y_+})^2,$$
$$(y_+^3-A)^3(1-Ay_+)=(1-Ay_+^3)^3(y_+-A),$$
$$(A^3-A)y_+^{10}+(1-A^4)y_+^9+3(A^3-A)y_+^6+3(A-A^3)y_+^4+(A^4-1)y_++(A-A^3)=0.$$
Dividing both sides of the equation by $A^2-1\neq 0$ and $y_+^5\neq 0$, we obtain
$$A(y_+^5-\frac{1}{y_+^5})-(1+A^2)(y_+^4-\frac{1}{y_+^4})+3A(y_+ -\frac{1}{y_+})=0.$$
Dividing both sides of the equation by $(y_+-\frac{1}{y_+})\neq 0$, we obtain
$$A(y_+^4+y_+^2+1+\frac{1}{y_+^2}+\frac{1}{y_+^4})-(1+A^2)(y_+^3+y_++\frac{1}{y_+}+\frac{1}{y_+^3})+3A=0.$$
Using $t=y_+ +\frac{1}{y_+}=y_+ + y_-$, we have
$$y_+^2+\frac{1}{y_+^2}=t^2-2,\quad y_+^3+\frac{1}{y_+^3}=t^3-3t,\quad y_+^4+\frac{1}{y_+^4}=t^4-4t^2+2$$
and obtain
$$At^4-(1+A^3)t^3-3At^2+2(1+A^2)t+4A=0.$$
Dividing by $t\neq 0$, we get
$$A(t^2+\frac{4}{t^2})-(1+A^2)(t-\frac{2}{t})-3A=0.$$
Using $d:=t-\frac{2}{t}$, we have $t^2+\frac{4}{t^2}=d^2+4$ and obtain
$$Ad^2-(1+A^2)d+A=0,\qquad \mbox{i.e. $d=A$, $d=\frac{1}{A}$.}$$
If $d=A$, then
\begin{eqnarray}
t^2-At-2=0.\label{QQ1}
\end{eqnarray}
But it is in a contradiction with
\begin{eqnarray}
At^3-(B^2+A^2+1)t^2+4B^2=0.\label{QQ2}
\end{eqnarray}
Really, multiplying both sides of (\ref{QQ1}) by $At$, we get
\begin{eqnarray}
At^3-A^2t^2-2At=0.\label{QQ3}
\end{eqnarray}
From (\ref{QQ2}) and (\ref{QQ3}), we obtain
\begin{eqnarray}
(1+B^2)t^2-2At-4B^2=0.\label{QQ4}
\end{eqnarray}
From (\ref{QQ4}) and (\ref{QQ1}), we get
$$\frac{2At+4B^2}{1+B^2}=At+2,\quad tA(1-B^2)=2(1-B^2),\quad t=\frac{2}{A},$$
because $B\neq \pm 1$. Substituting $t=\frac{2}{A}$ into (\ref{QQ1}), we get $\frac{4}{A^2}=4$, i.e. a contradiction, because $A\neq \pm 1$.

If $d=A^{-1}$, then
\begin{eqnarray}
At^2-t-2A=0.\label{Q1}
\end{eqnarray}
But it is in a contradiction with
\begin{eqnarray}
At^3-(B^2+A^2+1)t^2+4B^2=0.\label{Q2}
\end{eqnarray}
Really, multiplying both sides of (\ref{Q1}) by $t$, we get
\begin{eqnarray}
At^3-t^2-2t=0.\label{Q3}
\end{eqnarray}
From (\ref{Q2}) and (\ref{Q3}), we obtain
\begin{eqnarray}
(A^2+B^2)t^2-2At-4B^2=0.\label{Q4}
\end{eqnarray}
From (\ref{Q4}) and (\ref{Q1}), we get
$$\frac{2At+4B^2}{A^2+B^2}=\frac{t+2A}{A},\quad t(A^2-B^2)=2A(A^2-B^2),\quad t=2A,$$
because $A\neq \pm B$. Substituting $t=2A$ into (\ref{Q1}), we get $4A(A^2-1)=0$, i.e. a contradiction, because $A\neq 0$, $A\neq \pm 1$.

The lemma is proved.


\begin{thebibliography}{00}
\bibitem{Actor} A. Actor, Classical solutions of SU(2) Yang-Mills theories, \textit{Rev. Mod. Phys.} \textbf{51} (1979) 461--525.
\bibitem{deA} V. de Alfaro, S. Fubini, G. Furlan, A new classical solution of the Yang-Mills field equations, \textit{Phys. Lett. B} \textbf{65} (1976) 163.
\bibitem{ADHM} M. Atiyah, V. Drinfeld, N. Hitchin, Yu. Manin, Construction of instantons, \textit{Physics Letters A} \textbf{65} (1978) 185--187.
\bibitem{Bel} A. A. Belavin, A. M.  Polyakov, A. S. Schwartz, Yu. S. Tyupkin, Pseudoparticle solutions of the Yang-Mills equations, \textit{Phys. Lett. B} \textbf{59} (1975) 85.
\bibitem{Fad} L. D. Faddeev, A. A. Slavnov, \textit{Gauge Fields: An Introduction to Quantum Theory} (2nd ed., CRC Press, 2018).
\bibitem{SVD1} G. E. Forsythe, M. A. Malcolm, C. B. Moler, \textit{Computer Methods for Mathematical Computations} (Prentice Hall, Upper Saddle River, 1977).
\bibitem{SVD2} G. Golub, C. Van Loan, \textit{Matrix Computations} (3rd ed., Johns Hopkins University Press Baltimore, MD, USA, 1996).
\bibitem{GN} A. Gorsky, N. Nekrasov, Hamiltonian systems of Calogero type and two dimensional Yang-Mills theory, \textit{Nucl.Phys. B} \textbf{414} (1994) 213--238.
\bibitem{Gr} J. P. Greensite, Calculation of the Yang-Mills vacuum wave functional, \textit{Nuclear Physics B} \textbf{158} (1979) 469--496.
\bibitem{tH} G. 't Hooft, Magnetic Monopoles in Unified Gauge Theories, \textit{Nucl.Phys. B} \textbf{79} (1974) 276--284.
\bibitem{Ja} R. Jackiw, C. Rebbi, Vacuum Periodicity in a Yang-Mills Quantum Theory, \textit{Phys. Rev. Lett.} \textbf{37} (1976) 172.
\bibitem{YMM} N. G. Marchuk, On a field equation generating a new class of particular solutions to the Yang-Mills equations, \textit{Tr. Mat. Inst. Steklova} \textbf{285} (2014) 207--220 [\textit{Proceedings of the Steklov Institute of Mathematics} \textbf{285} (2014) 197--210.]
\bibitem{YMP} N. G. Marchuk, D. S. Shirokov, Constant Solutions of Yang-Mills Equations and Generalized Proca Equations, \textit{J. Geom. Symmetry Phys.} \textbf{42} (2016) 53--72.
\bibitem{PFE} N. G. Marchuk, D. S. Shirokov, General solutions of one class of field equations, \textit{Rep. Math. Phys.} \textbf{78}:3 (2016) 305--326.
\bibitem{nian} J. Nian, Y. Qian,  A topological way of finding solutions to Yang-Mills equations, arXiv:1901.06818 [hep-th].
\bibitem{Pol} A. M. Polyakov, Isomeric states of quantum fields, \textit{Sov.Phys. - JETP} \textbf{41} (1975) 988--995.
\bibitem{Sch} R. Schimming, On constant solutions of the Yang-Mills equations, \textit{Arch. Math.} \textbf{24}:2 (1988) 65--73.
\bibitem{Sch2} R. Schimming, E. Mundt, Constant potential solutions of the Yang-Mills equation, \textit{J. Math. Phys.} \textbf{33} (1992) 4250.
\bibitem{YMSh} D. S. Shirokov, Covariantly constant solutions of the Yang-Mills equations, \textit{Advances in Applied Clifford Algebras} \textbf{28} (2018) 53.
\bibitem{Wit} E. Witten, Some Exact Multipseudoparticle Solutions of Classical Yang-Mills Theory, \textit{Phys. Rev. Lett.} \textbf{38} (1977) 121.
\bibitem{WYa} T. T. Wu, C. N. Yang, \textit{in Properties of Matter Under Unusual Conditions}, edited by H. Mark and S. Fernbach (Interscience New York, 1968).
\bibitem{Zhdanov} R. Z. Zhdanov, V. I. Lahno, Symmetry and Exact Solutions of the Maxwell and SU(2) Yang-Mills Equations, \textit{Adv. Chem. Phys. Modern Nonlinear Optics} \textbf{119} part II (2001) 269--352.
\end{thebibliography}
\end{document}